\documentclass[11pt]{article}
\usepackage{amsmath, amsthm, amsfonts, amssymb, fullpage, indentfirst, cite}
\theoremstyle{plain}
\newtheorem{proposition}{Proposition}

\newtheorem{theorem}{Theorem}
\newtheorem{corollary}{Corollary}
\theoremstyle{definition}
\newtheorem{definition}{Definition}

\newtheorem{assumption}{Assumption}

\theoremstyle{remark}
\newtheorem{remark}{Remark}
\newtheorem{example}{Example}

\begin{document}

\title{\Large\bf Zero-Gradient-Sum Algorithms for Distributed Convex Optimization: The Continuous-Time Case\footnote{This work was supported by the National Science Foundation under grant CMMI-0900806.}}
\author{Jie Lu and Choon Yik Tang\\ School of Electrical and Computer Engineering\\ University of Oklahoma, Norman, OK 73019, USA\\ {\sf\{jie.lu-1,cytang\}@ou.edu}}
\date{\today}
\maketitle

\begin{abstract}
This paper presents a set of continuous-time distributed algorithms that solve unconstrained, separable, convex optimization problems over undirected networks with fixed topologies. The algorithms are developed using a Lyapunov function candidate that exploits convexity, and are called {\em Zero-Gradient-Sum} (ZGS) algorithms as they yield nonlinear networked dynamical systems that evolve invariantly on a zero-gradient-sum manifold and converge asymptotically to the unknown optimizer. We also describe a systematic way to construct ZGS algorithms, show that a subset of them actually converge exponentially, and obtain lower and upper bounds on their convergence rates in terms of the network topologies, problem characteristics, and algorithm parameters, including the algebraic connectivity, Laplacian spectral radius, and function curvatures. The findings of this paper may be regarded as a natural generalization of several well-known algorithms and results for distributed consensus, to distributed convex optimization.
\end{abstract}

\section{Introduction}\label{sec:intro}

This paper addresses the problem of solving an unconstrained, separable, convex optimization problem over an $N$-node multi-hop network, where each node $i$ observes a convex function $f_i$, and all the $N$ nodes wish to determine an optimizer $x^*$ that minimizes the sum of the $f_i$'s, i.e.,
\begin{align}
x^*\in\operatornamewithlimits{arg\,min}_x\sum_{i=1}^Nf_i(x).\label{eq:x*=argminsumf}
\end{align}
The problem \eqref{eq:x*=argminsumf} arises in many emerging and future applications of multi-agent systems and wired/\linebreak[0]wireless/\linebreak[0]social networks, where agents or nodes often need to collaborate in order to jointly accomplish sophisticated tasks in decentralized and optimal fashions \cite{Rabbat04}.

To date, a family of discrete-time {\em subgradient algorithms}, aimed at solving problem \eqref{eq:x*=argminsumf} under general convexity assumptions, have been reported in the literature. These subgradient algorithms may be roughly classified into two groups. The first group of algorithms \cite{Rabbat04, Blatt07, Johansson07, Ram09} are {\em incremental} in nature, relying on the passing of an estimate of $x^*$ around the network to operate. The second group of algorithms \cite{Johansson08, Nedic09, Ram10} are {\em non-incremental}, relying instead on a combination of subgradient updates and linear consensus iterations to operate, although gossip-based updates have also been considered \cite{Ram09b}. For each of these algorithms, a number of convergence properties have been established, including the resulting error bounds, asymptotic convergence, and convergence rates.

In \cite{LuJ11}, we introduced two gossip-style, distributed asynchronous algorithms, referred to as {\em Pairwise Equalizing} (PE) and {\em Pairwise Bisectioning} (PB), which solve the scalar version of problem \eqref{eq:x*=argminsumf}, in a manner that is fundamentally different from the aforementioned subgradient algorithms (e.g., PE and PB do not try to move along the gradient, nor do they require the notion of a stepsize). In \cite{LuJ10b}, we showed that the two basic ideas behind PE---namely, the conservation of a certain gradient sum at zero and the use of a convexity-inspired Lyapunov function---can be extended, leading to {\em Controlled Hopwise Equalizing} (CHE), a distributed asynchronous algorithm that allows individual nodes to use potential drops in the value of the Lyapunov function to control, on their own, when to initiate an iteration, so that problem \eqref{eq:x*=argminsumf} may be solved efficiently. In both the papers \cite{LuJ11, LuJ10b}, problem \eqref{eq:x*=argminsumf} was studied in a discrete-time, asynchronous setting, and only the scalar version of it was considered.

In this paper, we address problem \eqref{eq:x*=argminsumf} from a continuous-time and multi-dimensional standpoint, building upon the two basic ideas behind PE. Specifically, assuming that each $f_i$ in \eqref{eq:x*=argminsumf} is twice continuously differentiable and strongly convex and using the same Lyapunov function candidate as the one for PE and CHE, we first derive a family of continuous-time distributed algorithms called {\em Zero-Gradient-Sum} (ZGS) algorithms, with which the states of the resulting nonlinear networked dynamical systems slide along an invariant, zero-gradient-sum manifold and converge asymptotically to the unknown minimizer $x^*$ in \eqref{eq:x*=argminsumf}. We then describe a systematic way to construct ZGS algorithms and prove that a subset of them are exponentially convergent. For this subset of algorithms, we also obtain lower and upper bounds on their convergence rates as functions of the network topologies, problem characteristics, and algorithm parameters, including the algebraic connectivity, Laplacian spectral radius, and curvatures of the $f_i$'s. As another contribution of this paper, we show that some of the existing continuous-time distributed consensus algorithms (e.g., \cite{Olfati-Saber03, Olfati-Saber04, Hatano05, RenW05, Spanos05c, Tahbaz-Salehi07b}) are special cases of ZGS algorithms and are, interestingly, just a slight modification away from solving any problem of the form \eqref{eq:x*=argminsumf}. In addition, the well-known result from \cite{Olfati-Saber04}, which says that the convergence rate of a linear consensus algorithm is characterized by the algebraic connectivity of the underlying graph, is a special case of Theorem~\ref{thm:ZGSexpconvlb} here.

\section{Preliminaries}\label{sec:prel}

A twice continuously differentiable function $f:\mathbb{R}^n\rightarrow\mathbb{R}$ is {\em locally strongly convex} if for any convex and compact set $D\subset\mathbb{R}^n$, there exists a constant $\theta>0$ such that the following equivalent conditions hold \cite{Goebel08, Nesterov04}:
\begin{align}
f(y)-f(x)-\nabla f(x)^T(y-x)&\ge\frac{\theta}{2}\|y-x\|^2,\quad\forall x,y\in D,\label{eq:fyfxnfxyx>=t2||yx||2}\displaybreak[0]\\
(\nabla f(y)-\nabla f(x))^T(y-x)&\ge\theta\|y-x\|^2,\quad\forall x,y\in D,\label{eq:nfynfxyx>=t||yx||2}\displaybreak[0]\\
\nabla^2f(x)&\ge\theta I_n,\quad\forall x\in D,\label{eq:n2fx>=tI}
\end{align}
where $\|\cdot\|$ denotes the Euclidean norm, $\nabla f:\mathbb{R}^n\rightarrow\mathbb{R}^n$ is the gradient of $f$, $\nabla^2f:\mathbb{R}^n\rightarrow\mathbb{R}^{n\times n}$ is the Hessian of $f$, $I_n\in\mathbb{R}^{n\times n}$ is the identity matrix, and $\ge$ denotes matrix inequality (i.e., $A\ge B$ means $A-B$ is a positive semidefinite matrix). The function $f$ is {\em strongly convex} if there exists a constant $\theta>0$ such that the equivalent conditions \eqref{eq:fyfxnfxyx>=t2||yx||2}--\eqref{eq:n2fx>=tI} hold for $D=\mathbb{R}^n$, in which case $\theta$ is called the {\em convexity parameter} of $f$ \cite{Nesterov04}. Finally, for any twice continuously differentiable function $f:\mathbb{R}^n\rightarrow\mathbb{R}$, any convex set $D\subset\mathbb{R}^n$, and any constant $\Theta>0$, the following conditions are equivalent \cite{Nesterov04, Boyd04}:
\begin{align}
f(y)-f(x)-\nabla f(x)^T(y-x)&\le\frac{\Theta}{2}\|y-x\|^2,\quad\forall x,y\in D,\label{eq:fyfxnfxyx<=T2||yx||2}\displaybreak[0]\\
(\nabla f(y)-\nabla f(x))^T(y-x)&\le\Theta\|y-x\|^2,\quad\forall x,y\in D,\label{eq:nfynfxyx<=T||yx||2}\displaybreak[0]\\
\nabla^2f(x)&\le\Theta I_n,\quad\forall x\in D.\label{eq:n2fx<=TI}
\end{align}

\section{Problem Formulation}\label{sec:probform}

Consider a multi-hop network consisting of $N\ge2$ nodes, connected by bidirectional links in a fixed topology. The network is modeled as a connected, undirected graph $\mathcal{G}=(\mathcal{V},\mathcal{E})$, where $\mathcal{V}=\{1,2,\ldots,N\}$ represents the set of $N$ nodes and $\mathcal{E}\subset\{\{i,j\}:i,j\in\mathcal{V},i\neq j\}$ represents the set of links. Any two nodes $i,j\in\mathcal{V}$ are one-hop neighbors and can communicate if and only if $\{i,j\}\in\mathcal{E}$. The set of one-hop neighbors of each node $i\in\mathcal{V}$ is denoted as $\mathcal{N}_i=\{j\in\mathcal{V}:\{i,j\}\in\mathcal{E}\}$, and the communications are assumed to be delay- and error-free, with no quantization.

Suppose each node $i\in\mathcal{V}$ observes a function $f_i:\mathbb{R}^n\rightarrow\mathbb{R}$ satisfying the following assumption:

\begin{assumption}\label{asm:fi}
For each $i\in\mathcal{V}$, the function $f_i$ is twice continuously differentiable, strongly convex with convexity parameter $\theta_i>0$, and has a locally Lipschitz Hessian $\nabla^2f_i$.
\end{assumption}

Suppose, upon observing the $f_i$'s, all the $N$ nodes wish to solve the following unconstrained, separable, convex optimization problem:
\begin{align}
\min_{x\in\mathbb{R}^n}F(x),\label{eq:minF}
\end{align}
where the objective function $F:\mathbb{R}^n\rightarrow\mathbb{R}$ is defined as $F(x)=\sum_{i\in\mathcal{V}}f_i(x)$. The proposition below shows that $F$ has a unique minimizer $x^*\in\mathbb{R}^n$, so that problem \eqref{eq:minF} is well-posed:

\begin{proposition}\label{pro:exisuniqx*}
With Assumption~\ref{asm:fi}, there exists a unique $x^*\in\mathbb{R}^n$ such that $F(x^*)\le F(x)$ $\forall x\in\mathbb{R}^n$ and $\nabla F(x^*)=0$.
\end{proposition}

\begin{proof}
See Theorem~6 in \cite{Bazaraa73}.
\end{proof}

Given the above network and problem, the aim of this paper is to devise a continuous-time distributed algorithm of the form
\begin{align}
\dot{x}_i(t)&=\varphi_i(x_i(t),\mathbf{x}_{\mathcal{N}_i}(t);f_i,\mathbf{f}_{\mathcal{N}_i}),\quad\forall t\ge0,\;\forall i\in\mathcal{V},\label{eq:xd=pxxff}\displaybreak[0]\\
x_i(0)&=\chi_i(f_i,\mathbf{f}_{\mathcal{N}_i}),\quad\forall i\in\mathcal{V},\label{eq:x0=cff}
\end{align}
where $t\ge0$ denotes time; $x_i(t)\in\mathbb{R}^n$ is a state representing node $i$'s estimate of the unknown minimizer $x^*$ at time $t$; $\mathbf{x}_{\mathcal{N}_i}(t)=(x_j(t))_{j\in\mathcal{N}_i}\in\mathbb{R}^{n|\mathcal{N}_i|}$ is a vector obtained by stacking $x_j(t)$ $\forall j\in\mathcal{N}_i$; $\mathbf{f}_{\mathcal{N}_i}=(f_j)_{j\in\mathcal{N}_i}:\mathbb{R}^n\rightarrow\mathbb{R}^{|\mathcal{N}_i|}$ is a function obtained by stacking $f_j$ $\forall j\in\mathcal{N}_i$; $\varphi_i:\mathbb{R}^n\times\mathbb{R}^{n|\mathcal{N}_i|}\rightarrow\mathbb{R}^n$ is a locally Lipschitz function of $x_i(t)$ and $\mathbf{x}_{\mathcal{N}_i}(t)$ governing the dynamics of $x_i(t)$, whose definition may depend on $f_i$ and $\mathbf{f}_{\mathcal{N}_i}$; $\chi_i\in\mathbb{R}^n$ is a constant determining the initial state $x_i(0)$, whose value may depend on $f_i$ and $\mathbf{f}_{\mathcal{N}_i}$; $|\cdot|$ denotes the cardinality of a set; and $x_i(t)$, $f_i$, $\varphi_i$, and $\chi_i$ are maintained in node $i$'s local memory. The goal of the algorithm \eqref{eq:xd=pxxff} and \eqref{eq:x0=cff} is to steer all the estimates $x_i(t)$'s asymptotically (or, better yet, exponentially) to the unknown $x^*$, i.e.,
\begin{align}
\lim_{t\rightarrow\infty}x_i(t)=x^*,\quad\forall i\in\mathcal{V},\label{eq:limx=x*}
\end{align}
enabling all the nodes to cooperatively solve problem \eqref{eq:minF}. Note that to realize \eqref{eq:xd=pxxff} and \eqref{eq:x0=cff}, for each $i\in\mathcal{V}$, every node $j\in\mathcal{N}_i$ must send node $i$ its $x_j(t)$ at each time $t$ if $\varphi_i$ does depend on $x_j(t)$, and its $f_j$ at time $t=0$ if $\varphi_i$ or $\chi_i$ does depend on $f_j$.

\begin{remark}\label{rem:pcdep}
As it turns out and will be shown in Section~\ref{sec:zerogradsumalgo}, each $\varphi_i$ and $\chi_i$ in \eqref{eq:xd=pxxff} and \eqref{eq:x0=cff} do not have to depend on $\mathbf{f}_{\mathcal{N}_i}$, so that the nodes do not have to exchange their $f_i$'s. We note that the algorithm PB in \cite{LuJ11} also exhibits this feature, but the algorithms PE in \cite{LuJ11} and CHE in \cite{LuJ10b} do not.
\end{remark}

\section{Zero-Gradient-Sum Algorithms}\label{sec:zerogradsumalgo}

In this section, we develop a family of algorithms that achieve the stated goal. To facilitate the development, we let $\mathbf{x}^*=(x^*,x^*,\ldots,x^*)\in\mathbb{R}^{nN}$ and $\mathbf{x}=(x_1,x_2,\ldots,x_N)\in\mathbb{R}^{nN}$ denote the minimizer and state vectors, respectively, and write the latter as $\mathbf{x}(t)=(x_1(t),x_2(t),\ldots,x_N(t))$ when we wish to emphasize time or view it as a state trajectory.

Consider a Lyapunov function candidate $V:\mathbb{R}^{nN}\rightarrow\mathbb{R}$, defined in terms of the observed $f_i$'s as
\begin{align}
V(\mathbf{x})=\sum_{i\in\mathcal{V}}f_i(x^*)-f_i(x_i)-\nabla f_i(x_i)^T(x^*-x_i).\label{eq:V=sumfx*fxnfxx*x}
\end{align}
Notice that $V$ in \eqref{eq:V=sumfx*fxnfxx*x} is continuously differentiable because of Assumption~\ref{asm:fi}, and that it satisfies $V(\mathbf{x}^*)=0$. Moreover, $V$ is positive definite with respect to $\mathbf{x}^*$ and is radially unbounded, which can be seen by noting that Assumption~\ref{asm:fi} and the first-order strong convexity condition \eqref{eq:fyfxnfxyx>=t2||yx||2} imply
\begin{align}
V(\mathbf{x})\ge\sum_{i\in\mathcal{V}}\frac{\theta_i}{2}\|x^*-x_i\|^2,\quad\forall\mathbf{x}\in\mathbb{R}^{nN},\label{eq:V>=sumt2||x*x||2}
\end{align}
and \eqref{eq:V>=sumt2||x*x||2} in turn implies $V(\mathbf{x})>0$ $\forall\mathbf{x}\ne\mathbf{x}^*$ and $V(\mathbf{x})\rightarrow\infty$ as $\|\mathbf{x}\|\rightarrow\infty$. Therefore, $V$ in \eqref{eq:V=sumfx*fxnfxx*x} is a legitimate Lyapunov function candidate, which may be used to derive algorithms that ensure \eqref{eq:limx=x*}.

Taking the time derivative of $V$ along the state trajectory $\mathbf{x}(t)$ of the system \eqref{eq:xd=pxxff} and calling it $\dot{V}:\mathbb{R}^{nN}\rightarrow\mathbb{R}$, we obtain
\begin{align}
\dot{V}(\mathbf{x}(t))=\sum_{i\in\mathcal{V}}(x_i(t)-x^*)^T\nabla^2f_i(x_i(t))\varphi_i(x_i(t),\mathbf{x}_{\mathcal{N}_i}(t);f_i,\mathbf{f}_{\mathcal{N}_i}),\quad\forall t\ge0.\label{eq:Vd=sumxx*n2fxpxxff}
\end{align}
Due to Assumption~\ref{asm:fi} and to each $\varphi_i$ being locally Lipschitz, $\dot{V}$ in \eqref{eq:Vd=sumxx*n2fxpxxff} is continuous. In addition, it yields $\dot{V}(\mathbf{x}^*)=0$. Hence, if the functions $\varphi_i$ $\forall i\in\mathcal{V}$ are such that $\dot{V}$ is negative definite with respect to $\mathbf{x}^*$, i.e.,
\begin{align}
\sum_{i\in\mathcal{V}}(x_i-x^*)^T\nabla^2f_i(x_i)\varphi_i(x_i,\mathbf{x}_{\mathcal{N}_i};f_i,\mathbf{f}_{\mathcal{N}_i})<0,\quad\forall\mathbf{x}\ne\mathbf{x}^*,\label{eq:sumxx*n2fxpxxff<0}
\end{align}
the system \eqref{eq:xd=pxxff} would have a unique equilibrium point at $\mathbf{x}^*$, which by the Barbashin-Krasovskii theorem \cite{Khalil02} would be globally asymptotically stable. Consequently, regardless of how the constants $\chi_i$ $\forall i\in\mathcal{V}$ in \eqref{eq:x0=cff} are chosen, the goal \eqref{eq:limx=x*} would be accomplished.

As it follows from the above, the challenge lies in finding $\varphi_i$ $\forall i\in\mathcal{V}$, which collectively satisfy \eqref{eq:sumxx*n2fxpxxff<0}. Such $\varphi_i$'s, however, may be difficult to construct because $x^*$ in \eqref{eq:sumxx*n2fxpxxff<0} is unknown to any of the nodes, i.e., $x^*$ depends on {\em every} $f_i$ via \eqref{eq:minF}, but $\varphi_i$ maintained by each node $i\in\mathcal{V}$ can {\em only} depend on $f_i$ and $\mathbf{f}_{\mathcal{N}_i}$. As a result, one cannot let the $\varphi_i$'s depend on $x^*$, such as letting $\varphi_i(x_i,\mathbf{x}_{\mathcal{N}_i};f_i,\mathbf{f}_{\mathcal{N}_i})=x^*-x_i$ $\forall i\in\mathcal{V}$, even though this particular choice guarantees \eqref{eq:sumxx*n2fxpxxff<0} (since each $\nabla^2f_i(x_i)$ is positive definite, by \eqref{eq:n2fx>=tI}). Given that the required $\varphi_i$'s are not readily apparent, instead of searching for them, below we present an alternative approach toward the goal \eqref{eq:limx=x*}, which uses the same $V$ and $\dot{V}$ as in \eqref{eq:V=sumfx*fxnfxx*x} and \eqref{eq:Vd=sumxx*n2fxpxxff}, but demands neither local nor global asymptotic stability.

To state the approach, we first introduce two definitions: let $\mathcal{A}\subset\mathbb{R}^{nN}$ represent the {\em agreement set} and $\mathcal{M}\subset\mathbb{R}^{nN}$ represent the {\em zero-gradient-sum manifold}, defined respectively as
\begin{align}
\mathcal{A}&=\{(y_1,y_2,\ldots,y_N)\in\mathbb{R}^{nN}:y_1=y_2=\cdots=y_N\},\label{eq:A=yinRnNy=y=y}\displaybreak[0]\\
\mathcal{M}&=\{(y_1,y_2,\ldots,y_N)\in\mathbb{R}^{nN}:\sum_{i\in\mathcal{V}}\nabla f_i(y_i)=0\},\label{eq:M=yinRnNsumnfy=0}
\end{align}
so that $\mathbf{x}\in\mathcal{A}$ if and only if all the $x_i$'s agree, and $\mathbf{x}\in\mathcal{M}$ if and only if the sum of all the gradients $\nabla f_i$'s, evaluated respectively at the $x_i$'s, is zero. Notice from \eqref{eq:A=yinRnNy=y=y} that $\mathbf{x}^*\in\mathcal{A}$, from \eqref{eq:M=yinRnNsumnfy=0} and Proposition~\ref{pro:exisuniqx*} that $\mathbf{x}^*\in\mathcal{M}$, and from all of them that $\mathbf{x}\in\mathcal{A}\cap\mathcal{M}\Rightarrow\mathbf{x}=\mathbf{x}^*$. Thus, $\mathcal{A}\cap\mathcal{M}=\{\mathbf{x}^*\}$. Also note from the continuity of each $\nabla f_i$ that $\mathcal{M}$ is closed and from the Implicit Function Theorem and the nonsingularity of each $\nabla^2f_i(x)$ $\forall x\in\mathbb{R}^n$ that $\mathcal{M}$ is indeed a manifold of dimension $n(N-1)$.

Having introduced $\mathcal{A}$ and $\mathcal{M}$, we now describe the approach, which is based on the following recognition: to attain the goal \eqref{eq:limx=x*}, condition \eqref{eq:sumxx*n2fxpxxff<0}---which ensures that {\em every} trajectory $\mathbf{x}(t)$ goes to $\mathbf{x}^*$---is sufficient but not necessary. Rather, all that is needed is a {\em single} trajectory $\mathbf{x}(t)$, along which $\dot{V}(\mathbf{x}(t))\le0$ $\forall t\ge0$ and $\lim_{t\rightarrow\infty}V(\mathbf{x}(t))=0$, since the latter implies \eqref{eq:limx=x*}. Recognizing this, we next derive three conditions on the $\varphi_i$'s and $\chi_i$'s in \eqref{eq:xd=pxxff} and \eqref{eq:x0=cff} that produce such a trajectory. Assume, for a moment, that the $\chi_i$'s dictating the initial state $\mathbf{x}(0)$ have been decided, so that we may focus on the $\varphi_i$'s that shape the trajectory $\mathbf{x}(t)$ leaving $\mathbf{x}(0)$. Observe that $\dot{V}$ in \eqref{eq:Vd=sumxx*n2fxpxxff} takes the form $\dot{V}(\mathbf{x}(t))=\Phi_1(\mathbf{x}(t))-x^{*T}\Phi_2(\mathbf{x}(t))$ $\forall t\ge0$, where $\Phi_1:\mathbb{R}^{nN}\rightarrow\mathbb{R}$ and $\Phi_2:\mathbb{R}^{nN}\rightarrow\mathbb{R}^n$. Thus, the unknown $x^*$---which may undesirably affect the sign of $\dot{V}(\mathbf{x}(t))$---can be eliminated by setting $\Phi_2(\mathbf{x})=0$ $\forall\mathbf{x}\in\mathbb{R}^{nN}$, i.e., by forcing the $\varphi_i$'s to satisfy
\begin{align}
\sum_{i\in\mathcal{V}}\nabla^2f_i(x_i)\varphi_i(x_i,\mathbf{x}_{\mathcal{N}_i};f_i,\mathbf{f}_{\mathcal{N}_i})=0,\quad\forall\mathbf{x}\in\mathbb{R}^{nN}.\label{eq:sumn2fxpxxff=0}
\end{align}
With this first condition \eqref{eq:sumn2fxpxxff=0}, $\dot{V}$ becomes free of $x^*$, reducing to
\begin{align}
\dot{V}(\mathbf{x}(t))=\sum_{i\in\mathcal{V}}x_i(t)^T\nabla^2f_i(x_i(t))\varphi_i(x_i(t),\mathbf{x}_{\mathcal{N}_i}(t);f_i,\mathbf{f}_{\mathcal{N}_i}),\quad\forall t\ge0.\label{eq:Vd=sumxn2fxpxxff}
\end{align}
Next, notice that whenever $\mathbf{x}(t)$ is in the agreement set $\mathcal{A}$, due to \eqref{eq:A=yinRnNy=y=y} and \eqref{eq:sumn2fxpxxff=0}, $\dot{V}(\mathbf{x}(t))$ in \eqref{eq:Vd=sumxn2fxpxxff} must vanish. However, whenever $\mathbf{x}(t)\notin\mathcal{A}$, there is no such restriction. Hence, any time $\mathbf{x}(t)\notin\mathcal{A}$, $\dot{V}(\mathbf{x}(t))$ can be made negative by forcing the $\varphi_i$'s to also satisfy
\begin{align}
\sum_{i\in\mathcal{V}}x_i^T\nabla^2f_i(x_i)\varphi_i(x_i,\mathbf{x}_{\mathcal{N}_i};f_i,\mathbf{f}_{\mathcal{N}_i})<0,\quad\forall\mathbf{x}\in\mathbb{R}^{nN}-\mathcal{A}.\label{eq:sumxn2fxpxxff<0}
\end{align}
With this additional, second condition \eqref{eq:sumxn2fxpxxff<0}, no matter what $x^*$ is, $\dot{V}(\mathbf{x}(t))\le0$ along $\mathbf{x}(t)$, with equality if and only if $\mathbf{x}(t)\in\mathcal{A}$. Finally, note that \eqref{eq:xd=pxxff} and \eqref{eq:sumn2fxpxxff=0} imply
\begin{align*}
\frac{d}{dt}\sum_{i\in\mathcal{V}}\nabla f_i(x_i(t))=\sum_{i\in\mathcal{V}}\nabla^2f_i(x_i(t))\dot{x}_i(t)=0,\quad\forall t\ge0,
\end{align*}
while \eqref{eq:limx=x*}, the continuity of each $\nabla f_i$, and Proposition~\ref{pro:exisuniqx*} imply
\begin{align*}
\lim_{t\rightarrow\infty}\sum_{i\in\mathcal{V}}\nabla f_i(x_i(t))=\sum_{i\in\mathcal{V}}\nabla f_i(\lim_{t\rightarrow\infty}x_i(t))=\sum_{i\in\mathcal{V}}\nabla f_i(x^*)=\nabla F(x^*)=0.
\end{align*}
The former says that by making the $\varphi_i$'s satisfy \eqref{eq:sumn2fxpxxff=0}, the gradient sum $\sum_{i\in\mathcal{V}}\nabla f_i(x_i(t))$ along $\mathbf{x}(t)$ would remain constant over time, while the latter says that to achieve $\lim_{t\rightarrow\infty}V(\mathbf{x}(t))=0$ or equivalently \eqref{eq:limx=x*}, this constant sum must be zero, i.e., $\sum_{i\in\mathcal{V}}\nabla f_i(x_i(t))=0$ $\forall t\ge0$. Therefore, in view of \eqref{eq:x0=cff}, the $\chi_i$'s must be such that
\begin{align}
\sum_{i\in\mathcal{V}}\nabla f_i(\chi_i(f_i,\mathbf{f}_{\mathcal{N}_i}))=0,\label{eq:sumnfcff=0}
\end{align}
yielding the third and final condition.

By imposing algebraic constraints on the $\varphi_i$'s and $\chi_i$'s, conditions \eqref{eq:sumn2fxpxxff=0}, \eqref{eq:sumxn2fxpxxff<0}, and \eqref{eq:sumnfcff=0} characterize a family of algorithms. This family of algorithms share a number of properties, including one that has a nice geometric interpretation: observe from \eqref{eq:sumnfcff=0}, \eqref{eq:x0=cff}, and \eqref{eq:M=yinRnNsumnfy=0} that $\mathbf{x}(0)\in\mathcal{M}$ and further from \eqref{eq:xd=pxxff} and \eqref{eq:sumn2fxpxxff=0} that $\mathbf{x}(t)\in\mathcal{M}$ $\forall t>0$. Thus, every algorithm in the family produces a nonlinear networked dynamical system, whose trajectory $\mathbf{x}(t)$ begins on, and slides along, the zero-gradient-sum manifold $\mathcal{M}$, making $\mathcal{M}$ a positively invariant set. Due to this geometric interpretation, these algorithms are referred to as follows:

\begin{definition}\label{def:ZGS}
A continuous-time distributed algorithm of the form \eqref{eq:xd=pxxff} and \eqref{eq:x0=cff} is said to be a {\em Zero-Gradient-Sum} (ZGS) algorithm if $\varphi_i$ $\forall i\in\mathcal{V}$ are locally Lipschitz and satisfy \eqref{eq:sumn2fxpxxff=0} and \eqref{eq:sumxn2fxpxxff<0}, and $\chi_i$ $\forall i\in\mathcal{V}$ satisfy \eqref{eq:sumnfcff=0}.
\end{definition}

The following theorem lists the properties shared by ZGS algorithms, showing that every one of them is capable of asymptotically driving $\mathbf{x}(t)$ to $\mathbf{x}^*$, solving problem \eqref{eq:minF}:

\begin{theorem}\label{thm:ZGSasymconv}
Consider the network modeled in Section~\ref{sec:probform} and the use of a ZGS algorithm described in Definition~\ref{def:ZGS}. Suppose Assumption~\ref{asm:fi} holds. Then: (i) there exists a unique solution $\mathbf{x}(t)$ $\forall t\ge0$ to \eqref{eq:xd=pxxff} and \eqref{eq:x0=cff}; (ii) $\mathbf{x}(t)\in\mathcal{M}$ $\forall t\ge0$; (iii) $\dot{V}(\mathbf{x}(t))\le0$ $\forall t\ge0$, with equality if and only if $\mathbf{x}(t)=\mathbf{x}^*$; (iv) $\lim_{t\rightarrow\infty}V(\mathbf{x}(t))=0$; and (v) $\lim_{t\rightarrow\infty}\mathbf{x}(t)=\mathbf{x}^*$, i.e., \eqref{eq:limx=x*} holds.
\end{theorem}

\begin{proof}
Since $\varphi_i$ $\forall i\in\mathcal{V}$ are locally Lipschitz, to prove~(i) it suffices to show that every solution $\mathbf{x}(t)$ of \eqref{eq:xd=pxxff} and \eqref{eq:x0=cff} lies entirely in a compact subset of $\mathbb{R}^{nN}$. To this end, let $\mathcal{B}(\mathbf{x}^*,r)\subset\mathbb{R}^{nN}$ denote the closed-ball of radius $r\in[0,\infty)$ centered at $\mathbf{x}^*$, i.e., $\mathcal{B}(\mathbf{x}^*,r)=\{\mathbf{y}\in\mathbb{R}^{nN}:\|\mathbf{y}-\mathbf{x}^*\|\le r\}$. Note from \eqref{eq:Vd=sumxx*n2fxpxxff}, \eqref{eq:sumn2fxpxxff=0}, and \eqref{eq:sumxn2fxpxxff<0} that $\dot{V}(\mathbf{x}(t))\le0$ along $\mathbf{x}(t)$. This, together with \eqref{eq:V>=sumt2||x*x||2}, implies that $V(\mathbf{x}(0))\ge V(\mathbf{x}(t))\ge\frac{\min_{i\in\mathcal{V}}\theta_i}{2}\|\mathbf{x}(t)-\mathbf{x}^*\|^2$ along $\mathbf{x}(t)$. Hence, $\mathbf{x}(t)\in\mathcal{B}(\mathbf{x}^*,\sqrt{\frac{2V(\mathbf{x}(0))}{\min_{i\in\mathcal{V}}\theta_i}})$ $\forall t\ge0$, ensuring~(i). Statement~(ii) has been proven in the paragraph before Definition~\ref{def:ZGS}. To verify~(iii), notice again from \eqref{eq:Vd=sumxx*n2fxpxxff}, \eqref{eq:sumn2fxpxxff=0}, and \eqref{eq:sumxn2fxpxxff<0} that $\dot{V}(\mathbf{x}(t))=0$ if and only if $\mathbf{x}(t)\in\mathcal{A}$. Due to~(ii) and to $\mathcal{A}\cap\mathcal{M}=\{\mathbf{x}^*\}$ shown earlier, (iii) holds. To prove~(iv) and~(v), we will apply LaSalle's invariance principle from Theorem~4.4 in \cite{Khalil02} to the dynamics \eqref{eq:xd=pxxff}. Let $\Omega=\mathcal{M}\cap\{\mathbf{y}\in\mathbb{R}^{nN}:V(\mathbf{y})\le V(\mathbf{x}(0))\}$. Notice that $\Omega$ is compact since $\mathcal{M}$ is closed and $V$ in \eqref{eq:V=sumfx*fxnfxx*x} is continuous and satisfies \eqref{eq:V>=sumt2||x*x||2}. Also note from \eqref{eq:M=yinRnNsumnfy=0}, \eqref{eq:xd=pxxff}, and \eqref{eq:sumn2fxpxxff=0} that $\mathcal{M}$ is positively invariant, and from \eqref{eq:Vd=sumxx*n2fxpxxff}, \eqref{eq:sumn2fxpxxff=0}, \eqref{eq:sumxn2fxpxxff<0}, $\mathcal{A}\cap\mathcal{M}=\{\mathbf{x}^*\}$, and $\mathbf{x}^*\in\Omega\subset\mathcal{M}$ that $\dot{V}(\mathbf{x})\le0$ $\forall\mathbf{x}\in\Omega$, with equality if and only if $\mathbf{x}=\mathbf{x}^*$. Thus, $\Omega$ is positively invariant as well. Moreover, the largest invariant set in $\{\mathbf{y}\in\Omega:\dot{V}(\mathbf{y})=0\}=\{\mathbf{x}^*\}$ is $\{\mathbf{x}^*\}$, since it must be nonempty. It follows from Theorem~4.4 in \cite{Khalil02} that every solution starting in $\Omega$ approaches $\mathbf{x}^*$ as $t\rightarrow\infty$, including $\mathbf{x}(t)$. Therefore, (v) holds and, by the continuity of $V$, (iv) follows.
\end{proof}

Having established Theorem~\ref{thm:ZGSasymconv}, we now present a systematic way to construct ZGS algorithms. First, to find $\chi_i$'s that meet condition \eqref{eq:sumnfcff=0}, consider the following proposition, which shows that each $f_i$ has a unique minimizer $x_i^*\in\mathbb{R}^n$:

\begin{proposition}\label{pro:exisuniqxi*}
With Assumption~\ref{asm:fi}, for each $i\in\mathcal{V}$, there exists a unique $x_i^*\in\mathbb{R}^n$ such that $f_i(x_i^*)\le f_i(x)$ $\forall x\in\mathbb{R}^n$ and $\nabla f_i(x_i^*)=0$.
\end{proposition}

\begin{proof}
See Theorem~6 in \cite{Bazaraa73}.
\end{proof}

Proposition~\ref{pro:exisuniqxi*} implies that $\sum_{i\in\mathcal{V}}\nabla f_i(x_i^*)=0$. Hence, \eqref{eq:sumnfcff=0} can be met by simply letting
\begin{align}
\chi_i(f_i,\mathbf{f}_{\mathcal{N}_i})=x_i^*,\quad\forall i\in\mathcal{V},\label{eq:cff=xi*}
\end{align}
which is permissible since every $x_i^*$ in \eqref{eq:cff=xi*} depends just on $f_i$. It follows that each node $i\in\mathcal{V}$ must solve a ``local'' convex optimization problem $\min_{x\in\mathbb{R}^n}f_i(x)$ for $x_i^*$ before time $t=0$, in order to execute \eqref{eq:x0=cff} and \eqref{eq:cff=xi*}. We note, however, that \eqref{eq:cff=xi*} is sufficient for ensuring \eqref{eq:sumnfcff=0} but not necessary.

Next, to generate locally Lipschitz $\varphi_i$'s that ensure conditions \eqref{eq:sumn2fxpxxff=0} and \eqref{eq:sumxn2fxpxxff<0}, notice that each $\varphi_i$ is premultiplied by $\nabla^2f_i(x_i)$, which is nonsingular $\forall x_i\in\mathbb{R}^n$. Therefore, the impact of each $\nabla^2f_i(x_i)$ can be absorbed by setting
\begin{align}
\varphi_i(x_i,\mathbf{x}_{\mathcal{N}_i};f_i,\mathbf{f}_{\mathcal{N}_i})=(\nabla^2f_i(x_i))^{-1}\phi_i(x_i,\mathbf{x}_{\mathcal{N}_i};f_i,\mathbf{f}_{\mathcal{N}_i}),\quad\forall i\in\mathcal{V},\label{eq:pxxff=n2fxinvpxxff}
\end{align}
where $\phi_i:\mathbb{R}^n\times\mathbb{R}^{n|\mathcal{N}_i|}\rightarrow\mathbb{R}^n$ is a locally Lipschitz function of $x_i$ and $\mathbf{x}_{\mathcal{N}_i}$ maintained by node $i$. For each $i\in\mathcal{V}$, because $\nabla^2f_i$ is locally Lipschitz (due to Assumption~\ref{asm:fi}) and the determinant of $\nabla^2f_i(x_i)$ for every $x_i\in\mathbb{R}^n$ is no less than a positive constant $\theta_i^n$ (due further to \eqref{eq:n2fx>=tI}), the mapping $(\nabla^2f_i(\cdot))^{-1}:\mathbb{R}^n\rightarrow\mathbb{R}^{n\times n}$ in \eqref{eq:pxxff=n2fxinvpxxff} is locally Lipschitz. Thus, as long as the $\phi_i$'s are locally Lipschitz, so would the resulting $\varphi_i$'s, fulfilling the requirement. With \eqref{eq:pxxff=n2fxinvpxxff}, the dynamics \eqref{eq:xd=pxxff} become
\begin{align}
\dot{x}_i(t)=(\nabla^2f_i(x_i(t)))^{-1}\phi_i(x_i(t),\mathbf{x}_{\mathcal{N}_i}(t);f_i,\mathbf{f}_{\mathcal{N}_i}),\quad\forall t\ge0,\;\forall i\in\mathcal{V},\label{eq:xd=n2fxinvpxxff}
\end{align}
and conditions \eqref{eq:sumn2fxpxxff=0} and \eqref{eq:sumxn2fxpxxff<0} simplify to
\begin{align}
\sum_{i\in\mathcal{V}}\phi_i(x_i,\mathbf{x}_{\mathcal{N}_i};f_i,\mathbf{f}_{\mathcal{N}_i})&=0,\quad\forall\mathbf{x}\in\mathbb{R}^{nN},\label{eq:sumpxxff=0}\displaybreak[0]\\
\sum_{i\in\mathcal{V}}x_i^T\phi_i(x_i,\mathbf{x}_{\mathcal{N}_i};f_i,\mathbf{f}_{\mathcal{N}_i})&<0,\quad\forall\mathbf{x}\in\mathbb{R}^{nN}-\mathcal{A}.\label{eq:sumxpxxff<0}
\end{align}

Finally, to come up with locally Lipschitz $\phi_i$'s that assure conditions \eqref{eq:sumpxxff=0} and \eqref{eq:sumxpxxff<0}, suppose each $\phi_i$ is decomposed as
\begin{align}
\phi_i(x_i,\mathbf{x}_{\mathcal{N}_i};f_i,\mathbf{f}_{\mathcal{N}_i})=\sum_{j\in\mathcal{N}_i}\phi_{ij}(x_i,x_j;f_i,f_j),\quad\forall i\in\mathcal{V},\label{eq:pxxff=sumpxxff}
\end{align}
so that the dynamics \eqref{eq:xd=n2fxinvpxxff} become
\begin{align}
\dot{x}_i(t)=(\nabla^2f_i(x_i(t)))^{-1}\sum_{j\in\mathcal{N}_i}\phi_{ij}(x_i(t),x_j(t);f_i,f_j),\quad\forall t\ge0,\;\forall i\in\mathcal{V},\label{eq:xd=n2fxinvsumpxxff}
\end{align}
where $\phi_{ij}:\mathbb{R}^n\times\mathbb{R}^n\rightarrow\mathbb{R}^n$ is a locally Lipschitz function of $x_i$ and $x_j$ maintained by node $i$. Then, \eqref{eq:sumpxxff=0} can be ensured by requiring that every $\phi_{ij}$ and $\phi_{ji}$ pair be negative of each other, i.e.,
\begin{align}
\phi_{ij}(y,z;f_i,f_j)=-\phi_{ji}(z,y;f_j,f_i),\quad\forall i\in\mathcal{V},\;\forall j\in\mathcal{N}_i,\;\forall y,z\in\mathbb{R}^n,\label{eq:pyzff=pzyff}
\end{align}
since $\sum_{i\in\mathcal{V}}\phi_i=\sum_{i\in\mathcal{V}}\sum_{j\in\mathcal{N}_i}\phi_{ij}=\sum_{\{i,j\}\in\mathcal{E}}\phi_{ij}+\phi_{ji}=0$. With \eqref{eq:pxxff=sumpxxff} and \eqref{eq:pyzff=pzyff}, the left-hand side of \eqref{eq:sumxpxxff<0} turns into
\begin{align}
\sum_{i\in\mathcal{V}}x_i^T\phi_i(x_i,\mathbf{x}_{\mathcal{N}_i};f_i,\mathbf{f}_{\mathcal{N}_i})=\frac{1}{2}\sum_{i\in\mathcal{V}}\sum_{j\in\mathcal{N}_i}(x_i-x_j)^T\phi_{ij}(x_i,x_j;f_i,f_j),\quad\forall\mathbf{x}\in\mathbb{R}^{nN}.\label{eq:sumxpxxff=12sumsumxxpxxff}
\end{align}
Because the graph $\mathcal{G}$ is connected, for any $\mathbf{x}\in\mathbb{R}^{nN}-\mathcal{A}$, there exist $i\in\mathcal{V}$ and $j\in\mathcal{N}_i$ such that $x_i-x_j$ in \eqref{eq:sumxpxxff=12sumsumxxpxxff} is nonzero. Hence, \eqref{eq:sumxpxxff<0} can be guaranteed by requiring the $\phi_{ij}$'s to also satisfy
\begin{align}
(y-z)^T\phi_{ij}(y,z;f_i,f_j)<0,\quad\forall i\in\mathcal{V},\;\forall j\in\mathcal{N}_i,\;\forall y,z\in\mathbb{R}^n,\;y\ne z.\label{eq:yzpyzff<0}
\end{align}
Note that if \eqref{eq:pyzff=pzyff} holds, then $\phi_{ij}$ satisfies the inequality in \eqref{eq:yzpyzff<0} if and only if $\phi_{ji}$ does. Therefore, every pair of neighboring nodes $i,j\in\mathcal{V}$ need only minimal coordination before time $t=0$ to realize the dynamics \eqref{eq:xd=n2fxinvsumpxxff}: only one of them, say, node $i$, needs to construct a $\phi_{ij}$ that satisfies the inequality in \eqref{eq:yzpyzff<0}, and the other, i.e., node $j$, only needs to make sure that $\phi_{ji}=-\phi_{ij}$.

Examples~\ref{exa:ZGSelem} and~\ref{exa:ZGSconv} below illustrate two concrete ways to construct $\phi_{ij}$'s that obey \eqref{eq:pyzff=pzyff} and \eqref{eq:yzpyzff<0}:

\begin{example}\label{exa:ZGSelem}
Let $\phi_{ij}(y,z;f_i,f_j)=(\psi_{ij1}(y_1,z_1),\psi_{ij2}(y_2,z_2),\ldots,\psi_{ijn}(y_n,z_n))$ $\forall i\in\mathcal{V}$ $\forall j\in\mathcal{N}_i$ $\forall y=(y_1,y_2,\ldots,y_n)\in\mathbb{R}^n$ $\forall z=(z_1,z_2,\ldots,z_n)\in\mathbb{R}^n$, where each $\psi_{ij\ell}:\mathbb{R}^2\rightarrow\mathbb{R}$ can be {\em any} locally Lipschitz function satisfying $\psi_{ij\ell}(y_\ell,z_\ell)=-\psi_{ji\ell}(z_\ell,y_\ell)$ and $(y_\ell-z_\ell)\psi_{ij\ell}(y_\ell,z_\ell)<0$ $\forall y_\ell\ne z_\ell$ (e.g., $\psi_{ij\ell}(y_\ell,z_\ell)=\tanh(z_\ell-y_\ell)$ or $\psi_{ij\ell}(y_\ell,z_\ell)=-\psi_{ji\ell}(z_\ell,y_\ell)=\frac{z_\ell-y_\ell}{1+y_\ell^2}$). Then, \eqref{eq:pyzff=pzyff} and \eqref{eq:yzpyzff<0} hold.\hfill$\blacksquare$
\end{example}

\begin{example}\label{exa:ZGSconv}
Let $\phi_{ij}(y,z;f_i,f_j)=\nabla g_{\{i,j\}}(z)-\nabla g_{\{i,j\}}(y)$ $\forall i\in\mathcal{V}$ $\forall j\in\mathcal{N}_i$ $\forall y,z\in\mathbb{R}^n$, where each $g_{\{i,j\}}:\mathbb{R}^n\rightarrow\mathbb{R}$ can be {\em any} twice continuously differentiable and locally strongly convex function associated with link $\{i,j\}\in\mathcal{E}$ (e.g., $g_{\{i,j\}}(y)=\frac{1}{2}y^TA_{\{i,j\}}y$ where $A_{\{i,j\}}\in\mathbb{R}^{n\times n}$ is any symmetric positive definite matrix, or $g_{\{i,j\}}(y)=f_i(y)+f_j(y)$ if the nodes do not mind exchanging their $f_i$'s). Then, \eqref{eq:pyzff=pzyff} and \eqref{eq:yzpyzff<0} hold.\hfill$\blacksquare$
\end{example}

Examples~\ref{exa:lineconsalgo} and~\ref{exa:weigaverconsalgo} below show that some of the continuous-time distributed consensus algorithms in the literature are special cases of ZGS algorithms. In addition, they are just a slight modification away from solving general unconstrained, separable, convex optimization problems:

\begin{example}\label{exa:lineconsalgo}
Consider the scalar (i.e., $n=1$) linear consensus algorithm $\dot{x}_i(t)=\sum_{j\in\mathcal{N}_i}a_{ij}(x_j(t)-x_i(t))$ $\forall t\ge0$ $\forall i\in\mathcal{V}$ with symmetric parameters $a_{ij}=a_{ji}>0$ $\forall\{i,j\}\in\mathcal{E}$ and arbitrary initial states $x_i(0)=y_i$ $\forall i\in\mathcal{V}$, studied in \cite{Olfati-Saber04, Hatano05, RenW05, Tahbaz-Salehi07b}. By Definition~\ref{def:ZGS} and Theorem~\ref{thm:ZGSasymconv}, this algorithm is a ZGS algorithm that solves problem \eqref{eq:minF} for $f_i(x)=\frac{1}{2}(x-y_i)^2$ $\forall i\in\mathcal{V}$. Moreover, the algorithm is only a Hessian inverse and an initial condition away (i.e., $\dot{x}_i(t)=(\nabla^2f_i(x_i(t)))^{-1}\sum_{j\in\mathcal{N}_i}a_{ij}(x_j(t)-x_i(t))$ with $x_i(0)=x_i^*$) from solving {\em any} convex optimization problem of the form \eqref{eq:minF} for any $n\ge1$. Note that the same can be said about the scalar nonlinear consensus protocol in \cite{Olfati-Saber03}.\hfill$\blacksquare$
\end{example}

\begin{example}\label{exa:weigaverconsalgo}
Consider the multivariable (i.e., $n\ge1$) weighted-average consensus algorithm $\dot{x}_i(t)=W_i^{-1}\sum_{j\in\mathcal{N}_i}(x_j(t)-x_i(t))$ $\forall t\ge0$ $\forall i\in\mathcal{V}$ with $W_i=W_i^T>0$ and $x_i(0)=y_i$, proposed in \cite{Spanos05c} as a step toward a distributed Kalman filter. This algorithm is a ZGS algorithm that solves problem \eqref{eq:minF} for $f_i(x)=\frac{1}{2}(x-y_i)^TW_i(x-y_i)$ $\forall i\in\mathcal{V}$. Indeed, it is only a replacement of $W_i^{-1}$ by $(\nabla^2f_i(x_i(t)))^{-1}$ and $x_i(0)=y_i$ by $x_i(0)=x_i^*$ away from solving for general $f_i$'s.\hfill$\blacksquare$
\end{example}

\section{Convergence Rate Analysis}\label{sec:convrateanal}

In this section, we derive lower and upper bounds on the exponential convergence rates of the ZGS algorithms described in \eqref{eq:xd=n2fxinvsumpxxff} and Example~\ref{exa:ZGSconv}, i.e.,
\begin{align}
\dot{x}_i(t)=(\nabla^2f_i(x_i(t)))^{-1}\sum_{j\in\mathcal{N}_i}\nabla g_{\{i,j\}}(x_j(t))-\nabla g_{\{i,j\}}(x_i(t)),\quad\forall t\ge0,\;\forall i\in\mathcal{V},\label{eq:xd=n2fxinvsumngxngx}
\end{align}
which form a subset of those in Definition~\ref{def:ZGS}, but include the ones in Examples~\ref{exa:lineconsalgo} and~\ref{exa:weigaverconsalgo} as a subset. To enable the derivation, suppose an initial state $\mathbf{x}(0)\in\mathcal{M}$ is given (e.g., $\mathbf{x}(0)=(x_1^*,x_2^*,\ldots,x_N^*)$ as in \eqref{eq:x0=cff} and \eqref{eq:cff=xi*}). With this $\mathbf{x}(0)$, let $\mathcal{C}_i=\{x\in\mathbb{R}^n:f_i(x^*)-f_i(x)-\nabla f_i(x)^T(x^*-x)\le V(\mathbf{x}(0))\}$ $\forall i\in\mathcal{V}$ and let $\mathcal{C}=\operatorname{conv}\cup_{i\in\mathcal{V}}\mathcal{C}_i$, where $\operatorname{conv}$ denotes the convex hull. It follows from Assumption~\ref{asm:fi}, \eqref{eq:fyfxnfxyx>=t2||yx||2}, \eqref{eq:V=sumfx*fxnfxx*x}, and~(iii) in Theorem~\ref{thm:ZGSasymconv} that $\mathcal{C}_i$ $\forall i\in\mathcal{V}$ are compact, $\mathcal{C}$ is convex and compact, and
\begin{align}
x_i(t),x^*\in\mathcal{C}_i\subset\mathcal{C},\quad\forall t\ge0,\;\forall i\in\mathcal{V}.\label{eq:xx*inCsubsetC}
\end{align}
For each $i\in\mathcal{V}$, due to Assumption~\ref{asm:fi}, \eqref{eq:n2fx>=tI}, and $\mathcal{C}$ being compact, there exists a $\Theta_i\ge\theta_i$ such that
\begin{align}
\nabla^2f_i(x)\le\Theta_iI_n,\quad\forall x\in\mathcal{C}.\label{eq:n2fx<=TI2}
\end{align}
Moreover, for each $\{i,j\}\in\mathcal{E}$, due to \eqref{eq:nfynfxyx>=t||yx||2}, $g_{\{i,j\}}$ being locally strongly convex, and $\mathcal{C}$ being convex and compact, there exists a $\gamma_{\{i,j\}}>0$ such that
\begin{align}
(\nabla g_{\{i,j\}}(y)-\nabla g_{\{i,j\}}(x))^T(y-x)\ge\gamma_{\{i,j\}}\|y-x\|^2,\quad\forall x,y\in\mathcal{C}.\label{eq:ngyngxyx>=g||yx||2}
\end{align}
Furthermore, for each $\{i,j\}\in\mathcal{E}$, due to \eqref{eq:nfynfxyx>=t||yx||2}, \eqref{eq:n2fx>=tI}, \eqref{eq:ngyngxyx>=g||yx||2}, $\nabla^2g_{\{i,j\}}$ being continuous, and $\mathcal{C}$ being convex and compact, there exists a $\Gamma_{\{i,j\}}\ge\gamma_{\{i,j\}}$ such that
\begin{align}
\nabla^2g_{\{i,j\}}(x)\le\Gamma_{\{i,j\}}I_n,\quad\forall x\in\mathcal{C}.\label{eq:n2gx<=GI}
\end{align}
Observe that the constants $\Theta_i$'s, $\gamma_{\{i,j\}}$'s, and $\Gamma_{\{i,j\}}$'s---unlike the convexity parameters $\theta_i$'s---depend on the initial state $\mathbf{x}(0)$ via the sets $\mathcal{C}$ and $\mathcal{C}_i$'s. Thus, the convergence rate results obtained below are dependent on $\mathbf{x}(0)$ in general. One exception is the case where the $f_i$'s and $g_{\{i,j\}}$'s are quadratic functions, for which the $\theta_i$'s, $\Theta_i$'s, $\gamma_{\{i,j\}}$'s, and $\Gamma_{\{i,j\}}$'s may be taken as the smallest and largest eigenvalues of the Hessians of the $f_i$'s and $g_{\{i,j\}}$'s, respectively, independent of $\mathbf{x}(0)$. Finally, for convenience, let $\theta=\min_{i\in\mathcal{V}}\theta_i$, $\Theta=\max_{i\in\mathcal{V}}\Theta_i$, $\gamma=\min_{\{i,j\}\in\mathcal{E}}\gamma_{\{i,j\}}$, and $\Gamma=\max_{\{i,j\}\in\mathcal{E}}\Gamma_{\{i,j\}}$.

The following theorem establishes the exponential convergence of the ZGS algorithms \eqref{eq:xd=n2fxinvsumngxngx} and provides a lower bound $\rho$ on their convergence rates, that they can do no worse than:

\begin{theorem}\label{thm:ZGSexpconvlb}
Consider the network modeled in Section~\ref{sec:probform} and the use of a ZGS algorithm described in \eqref{eq:xd=n2fxinvsumngxngx}. Suppose Assumption~\ref{asm:fi} holds. Then,
\begin{align}
V(\mathbf{x}(t))&\le V(\mathbf{x}(0))e^{-\rho t},\quad\forall t\ge0,\label{eq:V<=Vert}\displaybreak[0]\\
\sum_{i\in\mathcal{V}}\theta_i\|x_i(t)-x^*\|^2&\le\sum_{i\in\mathcal{V}}\Theta_i\|x_i(0)-x^*\|^2e^{-\rho t},\quad\forall t\ge0,\label{eq:sumt||xx*||2<=sumT||xx*||2ert}
\end{align}
where $\rho=\sup\{\varepsilon\in\mathbb{R}:\varepsilon P\le Q\}>0$, $P=[P_{ij}]\in\mathbb{R}^{N\times N}$ is a positive semidefinite matrix given by
\begin{align}
P_{ij}=\begin{cases}(\frac{1}{2}-\frac{1}{N})\Theta_i+\frac{1}{2N^2}\sum_{\ell\in\mathcal{V}}\Theta_\ell, & \text{if $i=j$},\\ -\frac{\Theta_i+\Theta_j}{2N}+\frac{1}{2N^2}\sum_{\ell\in\mathcal{V}}\Theta_\ell, & \text{otherwise},\end{cases}\label{eq:Pij=T}
\end{align}
and $Q=[Q_{ij}]\in\mathbb{R}^{N\times N}$ is a positive semidefinite matrix given by
\begin{align}
Q_{ij}=\begin{cases}\sum_{\ell\in\mathcal{N}_i}\gamma_{\{i,\ell\}}, & \text{if $i=j$},\\ -\gamma_{\{i,j\}}, & \text{if $\{i,j\}\in\mathcal{E}$},\\ 0, & \text{otherwise}.\end{cases}\label{eq:Qij=g}
\end{align}
\end{theorem}

\begin{proof}
Let $\eta(t)=\frac{1}{N}\sum_{j\in\mathcal{V}}x_j(t)$ $\forall t\ge0$. Due to \eqref{eq:xx*inCsubsetC} and the convexity of $\mathcal{C}$, $\eta(t)\in\mathcal{C}$. Moreover, by Proposition~\ref{pro:exisuniqx*}, $\sum_{i\in\mathcal{V}}f_i(x^*)=F(x^*)\le F(\eta(t))=\sum_{i\in\mathcal{V}}f_i(\eta(t))$. Observe from \eqref{eq:M=yinRnNsumnfy=0} and~(ii) in Theorem~\ref{thm:ZGSasymconv} that $\sum_{i\in\mathcal{V}}\nabla f_i(x_i(t))=0$. Thus, from \eqref{eq:V=sumfx*fxnfxx*x}, $V(\mathbf{x}(t))\le\sum_{i\in\mathcal{V}}f_i(\eta(t))-f_i(x_i(t))-\nabla f_i(x_i(t))^T(\eta(t)-x_i(t))$. It follows from \eqref{eq:fyfxnfxyx<=T2||yx||2}, \eqref{eq:n2fx<=TI}, \eqref{eq:n2fx<=TI2}, \eqref{eq:xx*inCsubsetC}, and \eqref{eq:Pij=T} that
\begin{align}
V(\mathbf{x}(t))\le\sum_{i\in\mathcal{V}}\frac{\Theta_i}{2}\|x_i(t)-\frac{1}{N}\sum_{j\in\mathcal{V}}x_j(t)\|^2=\mathbf{x}(t)^T(P\otimes I_n)\mathbf{x}(t),\quad\forall t\ge0,\label{eq:V<=xPIx}
\end{align}
where $\otimes$ denotes the Kronecker product. Next, using \eqref{eq:xd=n2fxinvsumngxngx}, \eqref{eq:xd=pxxff}, and \eqref{eq:Vd=sumxn2fxpxxff}, we can write
\begin{align}
\dot{V}(\mathbf{x}(t))=-\frac{1}{2}\sum_{i\in\mathcal{V}}\sum_{j\in\mathcal{N}_i}(x_j(t)-x_i(t))^T(\nabla g_{\{i,j\}}(x_j(t))-\nabla g_{\{i,j\}}(x_i(t))),\quad\forall t\ge0.\label{eq:Vd=12sumsumxxngxngx}
\end{align}
Therefore, from \eqref{eq:ngyngxyx>=g||yx||2}, \eqref{eq:xx*inCsubsetC}, and \eqref{eq:Qij=g},
\begin{align}
-\dot{V}(\mathbf{x}(t))\ge\frac{1}{2}\sum_{i\in\mathcal{V}}\sum_{j\in\mathcal{N}_i}\gamma_{\{i,j\}}\|x_j(t)-x_i(t)\|^2=\mathbf{x}(t)^T(Q\otimes I_n)\mathbf{x}(t),\quad\forall t\ge0.\label{eq:Vd>=xQIx}
\end{align}
To relate \eqref{eq:V<=xPIx} and \eqref{eq:Vd>=xQIx}, notice from \eqref{eq:Pij=T} and \eqref{eq:Qij=g} that both $P$ and $Q$ are symmetric with zero row sums. Also, $\forall y=(y_1,y_2,\ldots,y_N)\in\mathbb{R}^N$, $y^TPy=\sum_{i\in\mathcal{V}}\frac{\Theta_i}{2}(y_i-\frac{1}{N}\sum_{j\in\mathcal{V}}y_j)^2\ge0$ and $y^TQy=\frac{1}{2}\sum_{i\in\mathcal{V}}\sum_{j\in\mathcal{N}_i}\gamma_{\{i,j\}}(y_j-y_i)^2\ge0$, where the equalities hold if and only if $y_1=y_2=\cdots=y_N$. Hence, both $P$ and $Q$ are positive semidefinite with $N-1$ positive eigenvalues, one eigenvalue at $0$, and $(\frac{1}{\sqrt{N}},\frac{1}{\sqrt{N}},\ldots,\frac{1}{\sqrt{N}})$ being its corresponding eigenvector. It follows that there exists an orthogonal $W\in\mathbb{R}^{N\times N}$ with the first column being $(\frac{1}{\sqrt{N}},\frac{1}{\sqrt{N}},\ldots,\frac{1}{\sqrt{N}})$, such that $W^TPW=\operatorname{diag}(0,\bar{P})$ and $W^TQW=\operatorname{diag}(0,\bar{Q})$, where $\bar{P},\bar{Q}\in\mathbb{R}^{(N-1)\times(N-1)}$, $\bar{P}=\bar{P}^T>0$, and $\bar{Q}=\bar{Q}^T>0$. Note that $\forall\varepsilon\in\mathbb{R}$, $\varepsilon P\le Q\Leftrightarrow\varepsilon\bar{P}\le\bar{Q}\Leftrightarrow\varepsilon I_{N-1}\le\bar{P}^{-1/2}\bar{Q}\bar{P}^{-1/2}$, where $\bar{P}^{1/2}=(\bar{P}^{1/2})^T>0$ is the square root of $\bar{P}$ via the spectral decomposition, i.e., $\bar{P}=\bar{P}^{1/2}\bar{P}^{1/2}$. Since $\rho=\sup\{\varepsilon\in\mathbb{R}:\varepsilon P\le Q\}$ and $\bar{P}^{-1/2}\bar{Q}\bar{P}^{-1/2}=(\bar{P}^{-1/2}\bar{Q}\bar{P}^{-1/2})^T>0$, $\rho$ is the smallest eigenvalue of $\bar{P}^{-1/2}\bar{Q}\bar{P}^{-1/2}$ which is positive and satisfies $\rho P\le Q$. Therefore, $\rho(P\otimes I_n)\le Q\otimes I_n$. This, along with \eqref{eq:V<=xPIx} and \eqref{eq:Vd>=xQIx}, implies $\rho V(\mathbf{x}(t))\le-\dot{V}(\mathbf{x}(t))$, i.e., \eqref{eq:V<=Vert}. Finally, due to \eqref{eq:fyfxnfxyx>=t2||yx||2}, \eqref{eq:V=sumfx*fxnfxx*x}, \eqref{eq:V<=Vert}, \eqref{eq:n2fx<=TI2}, \eqref{eq:n2fx<=TI}, \eqref{eq:fyfxnfxyx<=T2||yx||2}, and \eqref{eq:xx*inCsubsetC}, $\sum_{i\in\mathcal{V}}\frac{\theta_i}{2}\|x_i(t)-x^*\|^2\le V(\mathbf{x}(t))\le V(\mathbf{x}(0))e^{-\rho t}\le\sum_{i\in\mathcal{V}}\frac{\Theta_i}{2}\|x_i(0)-x^*\|^2e^{-\rho t}$, i.e., \eqref{eq:sumt||xx*||2<=sumT||xx*||2ert} holds.
\end{proof}

The lower bound $\rho$ in Theorem~\ref{thm:ZGSexpconvlb} can be calculated according to its proof: $\rho$ is the smallest eigenvalue of $\bar{P}^{-1/2}\bar{Q}\bar{P}^{-1/2}$. The corollary below gives another lower bound, which is not as tight as $\rho$ but is explicit in the algebraic connectivity $\lambda_2>0$ of the graph $\mathcal{G}$:

\begin{corollary}\label{cor:ZGSexpconvlb}
With the setup of Theorem~\ref{thm:ZGSexpconvlb},
\begin{align}
V(\mathbf{x}(t))&\le V(\mathbf{x}(0))e^{-\frac{2\gamma}{\Theta}\lambda_2t},\quad\forall t\ge0,\label{eq:V<=Ve2gTl2t}\displaybreak[0]\\
\|\mathbf{x}(t)-\mathbf{x}^*\|&\le\sqrt{\frac{\Theta}{\theta}}\|\mathbf{x}(0)-\mathbf{x}^*\|e^{-\frac{\gamma}{\Theta}\lambda_2t},\quad\forall t\ge0.\label{eq:||xx*||<=sqrtTt||xx*||egTl2t}
\end{align}
\end{corollary}

\begin{proof}
From \eqref{eq:V<=xPIx} and \eqref{eq:Vd>=xQIx},
\begin{align}
V(\mathbf{x}(t))&\le\frac{\Theta}{2}\sum_{i\in\mathcal{V}}\|x_i(t)-\frac{1}{N}\sum_{j\in\mathcal{V}}x_j(t)\|^2=\frac{\Theta}{2N}\mathbf{x}(t)^T(\mathcal{L}_{\bar{\mathcal{G}}}\otimes I_n)\mathbf{x}(t),\quad\forall t\ge0,\label{eq:V<=T2NxLGbIx}\displaybreak[0]\\
-\dot{V}(\mathbf{x}(t))&\ge\frac{\gamma}{2}\sum_{i\in\mathcal{V}}\sum_{j\in\mathcal{N}_i}\|x_j(t)-x_i(t)\|^2=\gamma\mathbf{x}(t)^T(\mathcal{L}_{\mathcal{G}}\otimes I_n)\mathbf{x}(t),\quad\forall t\ge0,\label{eq:Vd>=gxLGIx}
\end{align}
where $\mathcal{L}_{\bar{\mathcal{G}}}\in\mathbb{R}^{N\times N}$ is the Laplacian of the complete graph $\bar{\mathcal{G}}$ with vertex set $\mathcal{V}$, and $\mathcal{L}_{\mathcal{G}}\in\mathbb{R}^{N\times N}$ is the Laplacian of $\mathcal{G}$. Obviously, $\mathcal{L}_{\bar{\mathcal{G}}}$ has $N-1$ eigenvalues at $N$, $\mathcal{L}_{\mathcal{G}}$ has $N-1$ positive eigenvalues among which $\lambda_2$ is the smallest, and both $\mathcal{L}_{\bar{\mathcal{G}}}$ and $\mathcal{L}_{\mathcal{G}}$ have one eigenvalue at $0$ with $(\frac{1}{\sqrt{N}},\frac{1}{\sqrt{N}},\ldots,\frac{1}{\sqrt{N}})$ being its eigenvector. Let $W\in\mathbb{R}^{N\times N}$ contain $N$ orthonormal eigenvectors of $\mathcal{L}_{\mathcal{G}}$ in its columns. Then, $W^T\mathcal{L}_{\bar{\mathcal{G}}}W$ and $W^T\mathcal{L}_{\mathcal{G}}W$ are diagonal matrices similar to $\mathcal{L}_{\bar{\mathcal{G}}}$ and $\mathcal{L}_{\mathcal{G}}$, and both contain the eigenvalue $0$ in the same diagonal position. Hence, $\lambda_2W^T\mathcal{L}_{\bar{\mathcal{G}}}W\le NW^T\mathcal{L}_{\mathcal{G}}W$, so that $\lambda_2\mathcal{L}_{\bar{\mathcal{G}}}\le N\mathcal{L}_{\mathcal{G}}$. Applying this inequality to \eqref{eq:V<=T2NxLGbIx} and \eqref{eq:Vd>=gxLGIx}, we get $\frac{2\gamma}{\Theta}\lambda_2V(\mathbf{x}(t))\le-\dot{V}(\mathbf{x}(t))$, i.e., \eqref{eq:V<=Ve2gTl2t}. Finally, \eqref{eq:||xx*||<=sqrtTt||xx*||egTl2t} follows from \eqref{eq:V<=Ve2gTl2t} the same way \eqref{eq:sumt||xx*||2<=sumT||xx*||2ert} does from \eqref{eq:V<=Vert}.
\end{proof}

Notice that in the special case where $n=1$, $f_i(x)=\frac{1}{2}(x-x_i^*)^2$ $\forall i\in\mathcal{V}$, and $g_{\{i,j\}}(x)=\frac{1}{2}x^2$ $\forall\{i,j\}\in\mathcal{E}$, we may let the $\theta_i$'s, $\Theta_i$'s, and $\gamma_{\{i,j\}}$'s all be $1$. In this case, Theorem~\ref{thm:ZGSexpconvlb} and Corollary~\ref{cor:ZGSexpconvlb} both yield $\|\mathbf{x}(t)-\mathbf{x}^*\|\le\|\mathbf{x}(0)-\mathbf{x}^*\|e^{-\lambda_2t}$ $\forall t\ge0$, which coincides with the well-known convergence rate result for the linear consensus algorithm $\dot{x}_i(t)=\sum_{j\in\mathcal{N}_i}x_j(t)-x_i(t)$ $\forall t\ge0$ $\forall i\in\mathcal{V}$, reported in \cite{Olfati-Saber04}. Hence, Theorem~\ref{thm:ZGSexpconvlb} and Corollary~\ref{cor:ZGSexpconvlb} may be regarded as a generalization of such a result for distributed consensus, to distributed convex optimization.

The next theorem looks at the performance of the ZGS algorithms \eqref{eq:xd=n2fxinvsumngxngx} from the other end, providing an upper bound $\tilde{\rho}$ on their exponential convergence rates that mirrors Theorem~\ref{thm:ZGSexpconvlb}:

\begin{theorem}\label{thm:ZGSexpconvub}
Consider the network modeled in Section~\ref{sec:probform} and the use of a ZGS algorithm described in \eqref{eq:xd=n2fxinvsumngxngx}. Suppose Assumption~\ref{asm:fi} holds. Then,
\begin{align}
V(\mathbf{x}(t))&\ge V(\mathbf{x}(0))e^{-\tilde{\rho}t},\quad\forall t\ge0,\label{eq:V>=Vertt}\displaybreak[0]\\
\sum_{i\in\mathcal{V}}\Theta_i\|x_i(t)-x^*\|^2&\ge\sum_{i\in\mathcal{V}}\theta_i\|x_i(0)-x^*\|^2e^{-\tilde{\rho}t},\quad\forall t\ge0,\label{eq:sumT||xx*||2>=sumt||xx*||2ertt}
\end{align}
where $\tilde{\rho}=\inf\{\varepsilon\in\mathbb{R}:\varepsilon\tilde{P}\ge\tilde{Q}\}>0$, $\tilde{P}\in\mathbb{R}^{N\times N}$ is a positive definite matrix given by $\tilde{P}=\operatorname{diag}(\frac{\theta_1}{2},\frac{\theta_2}{2},\ldots,\frac{\theta_N}{2})$, and $\tilde{Q}=[\tilde{Q}_{ij}]\in\mathbb{R}^{N\times N}$ is a positive semidefinite matrix given by
\begin{align}
\tilde{Q}_{ij}=\begin{cases}\sum_{\ell\in\mathcal{N}_i}\Gamma_{\{i,\ell\}}, & \text{if $i=j$},\\ -\Gamma_{\{i,j\}}, & \text{if $\{i,j\}\in\mathcal{E}$},\\ 0, & \text{otherwise}.\end{cases}\label{eq:Qtij=G}
\end{align}
\end{theorem}

\begin{proof}
From \eqref{eq:fyfxnfxyx>=t2||yx||2} and \eqref{eq:V=sumfx*fxnfxx*x}, $V(\mathbf{x}(t))\ge\sum_{i\in\mathcal{V}}\frac{\theta_i}{2}\|x_i(t)-x^*\|^2=(\mathbf{x}(t)-\mathbf{x}^*)^T(\tilde{P}\otimes I_n)(\mathbf{x}(t)-\mathbf{x}^*)$ $\forall t\ge0$. From \eqref{eq:Vd=12sumsumxxngxngx}, \eqref{eq:n2gx<=GI}, \eqref{eq:n2fx<=TI}, \eqref{eq:nfynfxyx<=T||yx||2}, \eqref{eq:xx*inCsubsetC}, and \eqref{eq:Qtij=G}, $-\dot{V}(\mathbf{x}(t))\le\frac{1}{2}\sum_{i\in\mathcal{V}}\sum_{j\in\mathcal{N}_i}\Gamma_{\{i,j\}}\|x_j(t)-x_i(t)\|^2=\frac{1}{2}\sum_{i\in\mathcal{V}}\sum_{j\in\mathcal{N}_i}\Gamma_{\{i,j\}}\|(x_j(t)-x^*)-(x_i(t)-x^*)\|^2=(\mathbf{x}(t)-\mathbf{x}^*)^T(\tilde{Q}\otimes I_n)(\mathbf{x}(t)-\mathbf{x}^*)$ $\forall t\ge0$. Like $Q$ in \eqref{eq:Qij=g}, $\tilde{Q}$ in \eqref{eq:Qtij=G} is symmetric positive semidefinite with exactly one eigenvalue at $0$. Thus, so is $\tilde{P}^{-1/2}\tilde{Q}\tilde{P}^{-1/2}$, where $\tilde{P}^{1/2}=\operatorname{diag}(\sqrt{\frac{\theta_1}{2}},\sqrt{\frac{\theta_2}{2}},\ldots,\sqrt{\frac{\theta_N}{2}})$ is the square root of $\tilde{P}$. Since $\tilde{\rho}=\inf\{\varepsilon\in\mathbb{R}:\varepsilon\tilde{P}\ge\tilde{Q}\}$ and $\forall\varepsilon\in\mathbb{R}$, $\varepsilon\tilde{P}\ge\tilde{Q}\Leftrightarrow\varepsilon I_N\ge\tilde{P}^{-1/2}\tilde{Q}\tilde{P}^{-1/2}$, $\tilde{\rho}$ is the largest eigenvalue of $\tilde{P}^{-1/2}\tilde{Q}\tilde{P}^{-1/2}$ which is positive and such that $\tilde{\rho}\tilde{P}\ge\tilde{Q}$. Therefore, $\tilde{\rho}V(\mathbf{x}(t))\ge-\dot{V}(\mathbf{x}(t))$, proving \eqref{eq:V>=Vertt}. Finally, from \eqref{eq:V=sumfx*fxnfxx*x}, \eqref{eq:n2fx<=TI2}, \eqref{eq:n2fx<=TI}, \eqref{eq:fyfxnfxyx<=T2||yx||2}, \eqref{eq:xx*inCsubsetC}, \eqref{eq:V>=Vertt}, and \eqref{eq:fyfxnfxyx>=t2||yx||2}, we get \eqref{eq:sumT||xx*||2>=sumt||xx*||2ertt}.
\end{proof}

In contrast to $\rho$, the upper bound $\tilde{\rho}$ in Theorem~\ref{thm:ZGSexpconvub} is the largest eigenvalue of $\tilde{P}^{-1/2}\tilde{Q}\tilde{P}^{-1/2}$. The next corollary is to Theorem~\ref{thm:ZGSexpconvub} as Corollary~\ref{cor:ZGSexpconvlb} is to Theorem~\ref{thm:ZGSexpconvlb}, giving another upper bound that is not as tight as $\tilde{\rho}$ but is explicit in the spectral radius $\lambda_N>0$ of the graph Laplacian $\mathcal{L}_{\mathcal{G}}$:

\begin{corollary}\label{cor:ZGSexpconvub}
With the setup of Theorem~\ref{thm:ZGSexpconvub},
\begin{align}
V(\mathbf{x}(t))&\ge V(\mathbf{x}(0))e^{-\frac{2\Gamma}{\theta}\lambda_Nt},\quad\forall t\ge0,\label{eq:V>=Ve2GtlNt}\displaybreak[0]\\
\|\mathbf{x}(t)-\mathbf{x}^*\|&\ge\sqrt{\frac{\theta}{\Theta}}\|\mathbf{x}(0)-\mathbf{x}^*\|e^{-\frac{\Gamma}{\theta}\lambda_Nt},\quad\forall t\ge0.\label{eq:||xx*||>=sqrttT||xx*||eGtlNt}
\end{align}
\end{corollary}

\begin{proof}
From the proof of Theorem~\ref{thm:ZGSexpconvub}, $\forall t\ge0$, we have $V(\mathbf{x}(t))\ge\sum_{i\in\mathcal{V}}\frac{\theta}{2}\|x_i(t)-x^*\|^2=\frac{\theta}{2}\|\mathbf{x}(t)-\mathbf{x}^*\|^2$ and $-\dot{V}(\mathbf{x}(t))\le\frac{1}{2}\sum_{i\in\mathcal{V}}\sum_{j\in\mathcal{N}_i}\Gamma\|(x_j(t)-x^*)-(x_i(t)-x^*)\|^2=\Gamma(\mathbf{x}(t)-\mathbf{x}^*)^T(\mathcal{L}_{\mathcal{G}}\otimes I_n)(\mathbf{x}(t)-\mathbf{x}^*)\le\Gamma\lambda_N\|\mathbf{x}(t)-\mathbf{x}^*\|^2$. Consequently, $\frac{2\Gamma}{\theta}\lambda_NV(\mathbf{x}(t))\ge-\dot{V}(\mathbf{x}(t))$, implying that \eqref{eq:V>=Ve2GtlNt} and \eqref{eq:||xx*||>=sqrttT||xx*||eGtlNt} hold.
\end{proof}

Note that for the special case below Corollary~\ref{cor:ZGSexpconvlb}, we may let the $\Gamma_{\{i,j\}}$'s be $1$, so that Theorem~\ref{thm:ZGSexpconvub} and Corollary~\ref{cor:ZGSexpconvub} both lead to $\|\mathbf{x}(t)-\mathbf{x}^*\|\ge\|\mathbf{x}(0)-\mathbf{x}^*\|e^{-\lambda_Nt}$ $\forall t\ge0$, which is again known. Finally, note that the above analysis provides a framework for studying the interplay among network topologies (i.e., $\mathcal{V}$ and $\mathcal{E}$), problem characteristics (i.e., the $f_i$'s, $\theta_i$'s, and $\Theta_i$'s), and ZGS algorithm parameters (i.e., the $g_{\{i,j\}}$'s, $\gamma_{\{i,j\}}$'s, and $\Gamma_{\{i,j\}}$'s), which may be worthy of further research.

\section{Conclusion}\label{sec:concl}

In this paper, using a convexity-based Lyapunov function candidate, we have developed a set of continuous-time ZGS algorithms, which solve a class of distributed convex optimization problems over networks. We have established the asymptotic and exponential convergence of these algorithms and derived lower and upper bounds on their convergence rates. We have also shown that the ZGS algorithms for distributed convex optimization are closely related to the basic algorithms for distributed consensus, suggesting that the former may be extended in a number of directions just like the latter were, in ways that possibly parallel the latter.

\bibliographystyle{IEEEtran}
\bibliography{paper}

\end{document}